\documentclass[letterpaper, 10 pt, conference]{IEEEconf}
\IEEEoverridecommandlockouts  
\usepackage{graphicx} % Required for inserting images
\usepackage[T1]{fontenc}
\usepackage{lmodern}

\usepackage{amsmath}
\usepackage{tikz}
\usepackage{amsfonts}
\usepackage{nicefrac}
\usepackage{dsfont}
\let\labelindent\relax
\usepackage[shortlabels]{enumitem}
\usepackage{algorithm}
\usepackage{algpseudocode}
\newtheorem{theorem}{Theorem}
\newtheorem{definition}{Definition}
\newcommand{\expectation}{\mathds{E}}
\newcommand{\argmin}{\arg\min}

\newcommand{\R}{\mathds{R}}
\newcommand{\probabilitymeasure}{\mathds{P}}
\newcommand{\indicator}{\mathds{1}}

\newcommand{\timeindex}{k}
\newcommand{\alttimeindex}{m}

\newcommand{\statespace}{\mathcal{X}}
\newcommand{\statedim}{X}
\newcommand{\statevar}{x}
\newcommand{\stateidx}{i}

\newcommand{\textobservationspace}{\mathcal{Y}^\prime}
\newcommand{\observationspace}{\mathcal{Y}}
\newcommand{\llmobservation}{y^\prime}
\newcommand{\observationmatrix}{B}
\newcommand{\observation}{y}

\newcommand{\prior}{\pi}
\newcommand{\priorupdate}{\mathcal{T}}
\newcommand{\actionprob}{R}

\newcommand{\action}{u}
\newcommand{\decision}{a}
\newcommand{\actionalt}{\Tilde{u}}

\newcommand{\actionspace}{\mathcal{U}}
\newcommand{\cost}{c}

\newcommand{\bayesianagenttext}{language-driven decision-making agent}
\newcommand{\bayesianagentstext}{language-driven decision-making agents}
\newcommand{\bayesianagent}{LDDA}
\newcommand{\bayesianagents}{LDDAs}

\newcommand{\markovkernel}{P}
\newcommand{\herdtime}{K}
\newcommand{\policy}{\mu}
\newcommand{\avgcost}{J}
\newcommand{\discountfactor}{\rho}
\newcommand{\stoppingtime}{\tau}
\newcommand{\delaycost}{d}
\newcommand{\errorcost}{\delta}
\newcommand{\indicatorstate}{e}
\newcommand{\filtration}{\mathcal{F}}
\newcommand{\actionfiltration}{\mathcal{G}}

\newcommand{\probspace}{\mathcal{P}}
\newcommand{\decisionspace}{\mathcal{D}}
\newcommand{\region}{\mathcal{R}}
\newcommand{\timehorizon}{T}

\newcommand{\threshold}{\gamma}
\newcommand{\lines}{\mathcal{L}}
\newcommand{\linesymb}{L}
\newcommand{\bigcost}{C}
\newcommand{\priorratio}{\Lambda}
\newcommand{\stateidxalt}{j}
\newcommand{\actionratio}{\Gamma}
\newcommand{\constant}{\kappa}

\title{\bf Identifying Hate Speech Peddlers in Online Platforms. A Bayesian Social Learning Approach for Large Language Model Driven Decision-Makers }
\author{Adit Jain and Vikram Krishnamurthy, \textit{IEEE Fellow}% <-this % stops a space
\thanks{A. Jain and V. Krishnamurthy are with the Department of Electrical and Computer Engineering, 
        Cornell University, Ithaca, NY, 14853.
        Corresponding email: {\tt\small aj457@cornell.edu }}%
}

\begin{document}
%% 1.5 page introduction + Abstract
%% 1 page Social Learning framework
%% 0.5 page main result: herding in social learning
%% 1 page Methods 
%% 1 page numerical results
%% 0.5 pg references + conclusion + additional applications
%% Proofs: 0.5 pg
\maketitle
\begin{abstract}
    This paper studies the problem of autonomous agents performing Bayesian social learning for sequential detection when the observations of the state belong to a high-dimensional space and are expensive to analyze. Specifically, when the observations are textual, the Bayesian agent can use a large language model (LLM) as a map to get a low-dimensional private observation. The agent performs Bayesian learning and takes an action that minimizes the expected cost and is visible to subsequent agents. We prove that a sequence of such Bayesian agents herd in finite time to the public belief and take the same action disregarding the private observations. We propose a stopping time formulation for quickest time herding in social learning and optimally balance privacy and herding. Structural results are shown on the threshold nature of the optimal policy to the stopping time problem. We illustrate the application of our framework when autonomous Bayesian detectors aim to sequentially identify if a user is a hate speech peddler on an online platform by parsing text observations using an LLM.
    We numerically validate our results on real-world hate speech datasets. We show that autonomous Bayesian agents designed to flag hate speech peddlers in online platforms herd and misclassify the users when the public prior is strong. We also numerically show the effect of a threshold policy in delaying herding. 
\end{abstract}
\section{Introduction}
%% Describe the problem and proposed solution
In a social learning setting, Bayesian decision makers (or agents) sequentially detect and act based on observations of the state of nature and actions of the previous agents. The observations are often from a high-dimensional space and are expensive to analyze. If the observations are textual, the Bayesian decision-makers make use of tools like large language models as sensors to get a low-dimensional representation of the text observation~\footnote{Our results extend to multi-modal observations including a combination of images, text, and audios, and the LLM sensor can be replaced by an appropriate transformer network. }. The observations are private to the respective agents (due to privacy and communication constraints), but the actions are public. Based on the action of the current agent, the successive agents update their prior, which can be shown to cascade and remain fixed after some time. When a cascade happens, the agents herd and take the same action irrespective of their observation. 

\textit{How can herding in \bayesianagentstext\ (\bayesianagents) 
be controlled so that sequential detection is achieved with minimum privacy loss to the agents?}

We propose a stochastic control based approach for quickest time herding, which ensures that the \bayesianagents\ act benevolently based on the public prior. This framework is of interest when the \bayesianagents\ are deployed in a sensitive application and interact with humans, e.g., on online platforms where the task of the \bayesianagents\ is to sequentially flag a user as a hate speech peddler based on the text comments.  The main contributions of this paper are,
\subsection{Main Results}
\begin{enumerate}
    \item We formalize the problem of Bayesian agents sequentially performing social learning and taking public actions to minimize their expected cost.  The Bayesian agents have large language models (LLM) sensors to parse private text observations. The social learning protocol followed by a sequence of \bayesianagentstext\ (\bayesianagents) is derived. 
    \item For the described social learning protocol, Theorem~\ref{th:herding} shows that the public prior of the \bayesianagents\ forms an information cascade which causes the actions of the \bayesianagents\ to herd with probability 1, as a result of which the later \bayesianagents\ ignore their respective private observations.
    \item This paper formulates a stopping time problem to delay herding in the social learning performed by the \bayesianagents. We propose a socialistic cumulative cost and show in Theorem~\ref{th:threshold} that under structural assumptions, the optimal policy for this cost has a threshold switching curve w.r.t. the public belief. 
    \item The theoretical result on herding is numerically demonstrated on a classification task for hate speech peddlers on online platforms. The \bayesianagents\ equipped classify hate speech peddlers using text observations from a real-life hate speech dataset. We also study how a threshold policy can delay herding of \bayesianagents\ in the quickest time herding setup. 
\end{enumerate}
\subsection{Motivation}
\textit{Social Learning in Bayesian agents equipped with LLMs as sensors:}
The effectiveness of LLM, which are transformer neural networks with billions of parameters and are trained on trillions of tokens of textual data to parse long texts for summarizing, compiling key facts, and generating new text, has made their deployment widespread~\cite{min_recent_2023}. Many applications have been proposed in healthcare, online platform moderation, and finance, where these LLMs are used to parse the textual observations and make decisions based on their outputs~\cite{li_pre-trained_2022}. The output of the LLMs, either in a specified format or as embeddings, is often used as inputs to other Bayesian entities, including classifiers~\cite{minaee_deep_2021}. 

A single Bayesian agent who uses the LLM to parse text observations, update its Bayesian belief, and take action sequentially can be modeled as a series of Bayesian agents. Multiple Bayesian agents, each receiving private observations, is motivated by privacy, improved detection, and finite context length. If the same private observation (even the low-dimensional representation) is used, the LLM can be fine-trained on this data, which might contain sensitive information~\cite{duan_flocks_nodate,mohtashami_social_2024}. Also, different LLMs can be given a diverse set of contexts, which enables reducing the bias involved with their decisions~\cite{kaddour_challenges_2023}. 

Recent research has studied modeling LLMs as autonomous agents and making LLM part of bigger autonomous agents, including robots, self-driving cars, and programming co-pilots~\cite{cui_drive_2024,suri_software_2023}. Recently~\cite{mohtashami_social_2024} looked at social learning in LLMs using a teacher-student framework, but this work was in a static setting where the LLMs don't have a belief that they adaptively update. Sequential social learning in Bayesian agents has been studied extensively~\cite{chamley_rational_2004}, and our work formalizes the problem of Bayesian social learning in~\bayesianagents.

\begin{figure}[t!]
    \centering
    \vspace{2mm}    
    \resizebox{\columnwidth}{!}{
\tikzset{every picture/.style={line width=0.75pt}} %set default line width to 0.75pt        

\begin{tikzpicture}[x=0.75pt,y=0.75pt,yscale=-1,xscale=1]
%uncomment if require: \path (0,300); %set diagram left start at 0, and has height of 300

%Shape: Rectangle [id:dp0802837588204186] 
\draw   (127,95.5) -- (254,95.5) -- (254,153.5) -- (127,153.5) -- cycle ;
%Shape: Rectangle [id:dp46898761573759873] 
\draw   (298,76.5) -- (425,76.5) -- (425,193.5) -- (298,193.5) -- cycle ;
%Straight Lines [id:da7095275342668594] 
\draw    (253,123.5) -- (297,123.5) ;
\draw [shift={(299,123.5)}, rotate = 180] [color={rgb, 255:red, 0; green, 0; blue, 0 }  ][line width=0.75]    (10.93,-3.29) .. controls (6.95,-1.4) and (3.31,-0.3) .. (0,0) .. controls (3.31,0.3) and (6.95,1.4) .. (10.93,3.29)   ;
%Shape: Rectangle [id:dp9767758451848383] 
\draw   (107,55.5) -- (454,55.5) -- (454,219.5) -- (107,219.5) -- cycle ;
%Straight Lines [id:da2516040659759482] 
\draw    (425,123.5) -- (477,123.5) ;
\draw [shift={(479,123.5)}, rotate = 180] [color={rgb, 255:red, 0; green, 0; blue, 0 }  ][line width=0.75]    (10.93,-3.29) .. controls (6.95,-1.4) and (3.31,-0.3) .. (0,0) .. controls (3.31,0.3) and (6.95,1.4) .. (10.93,3.29)   ;
%Straight Lines [id:da5360919343437325] 
\draw    (85,125.5) -- (123,125.5) ;
\draw [shift={(125,125.5)}, rotate = 180] [color={rgb, 255:red, 0; green, 0; blue, 0 }  ][line width=0.75]    (10.93,-3.29) .. controls (6.95,-1.4) and (3.31,-0.3) .. (0,0) .. controls (3.31,0.3) and (6.95,1.4) .. (10.93,3.29)   ;
%Shape: Rectangle [id:dp5780769549332825] 
\draw   (306,89.5) -- (418,89.5) -- (418,125.5) -- (306,125.5) -- cycle ;
%Shape: Rectangle [id:dp10654315958759186] 
\draw   (305,142.5) -- (417,142.5) -- (417,178.5) -- (305,178.5) -- cycle ;
%Straight Lines [id:da6779313765403878] 
\draw    (356,44.5) -- (356,73.5) ;
\draw [shift={(356,75.5)}, rotate = 270] [color={rgb, 255:red, 0; green, 0; blue, 0 }  ][line width=0.75]    (10.93,-3.29) .. controls (6.95,-1.4) and (3.31,-0.3) .. (0,0) .. controls (3.31,0.3) and (6.95,1.4) .. (10.93,3.29)   ;

% Text Node
\draw (267,98.4) node [anchor=north west][inner sep=0.75pt]    {$\observation_\timeindex$};
% Text Node
% \draw (61,116.4) node [anchor=north west][inner sep=0.75pt]    {$y'_{k}$};
% Text Node
\draw (491,112.4) node [anchor=north west][inner sep=0.75pt]    {$\action_\timeindex$};
% Text Node
\draw (21,117.4) node [anchor=north west][inner sep=0.75pt]    {$\statevar_\timeindex \sim \llmobservation_\timeindex$};

% Text Node
\draw (142,110) node [anchor=north west][inner sep=0.75pt]   [align=left] {\begin{minipage}[lt]{77.03pt}\setlength\topsep{0pt}
\textbf{L}arge \textbf{L}anguage
\begin{center}
\textbf{M}odel
\end{center}

\end{minipage}};
% Text Node
% Text Node
\draw (125,226) node [anchor=north west][inner sep=0.75pt]   [align=left] {Large \textbf{L}anguage Model \textbf{D}riven \textbf{D}ecision Making \textbf{A}gent};
% Text Node
\draw (312,100) node [anchor=north west][inner sep=0.75pt][align=left] {\begin{minipage}[lt]{73.37pt}\setlength\topsep{0pt}
\begin{center}
Prior
\end{center}

\end{minipage}};
% Text Node
\draw (312,150) node [anchor=north west][inner sep=0.75pt][align=left] {\begin{minipage}[lt]{73.37pt}\setlength\topsep{0pt}
\begin{center}
Likelihood (NN)
\end{center}
\end{minipage}};
% Text Node
\draw (318,21.4) node [anchor=north west][inner sep=0.75pt]    {$(\action_{1},\dots,\action_{\timeindex-1})$};

\end{tikzpicture}
}
    \vspace{-10mm}
    \caption{The \bayesianagenttext\ (\bayesianagent) modeled in the text has two components: a) a large language model (LLM) which acts as a noisy sensor to provide a low-dimensional observation $\observation_\timeindex$ for the true state $\statevar_\timeindex$ by parsing the high-dimensional text observation $\llmobservation_\timeindex$ and b) a neural Bayesian engine which uses Bayes rules to update the belief about the state using \eqref{eq:bayesrule}. A neural network parameterizes the likelihood. The \bayesianagent\ outputs the action minimizing the expected cost \eqref{eq:action} and uses the actions of the previous \bayesianagents\ to update the prior using \eqref{eq:priorupdate}.  }
    \label{fig:singleagent}
    \vspace{-5mm}
\end{figure}
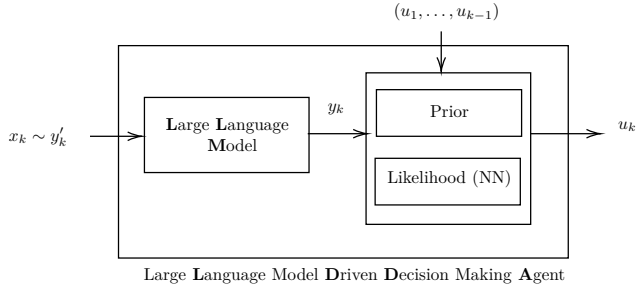

\textit{Hate Speech Peddler Identification On Social Networks: }
Identifying hate speech\footnote{There is an active debate on the definition of hate speech and the tradeoff between free speech and hate speech~\cite{howard_free_2019}. Hence, to circumvent this discussion, we use hate speech as an exemplary case study of our methods, and the definition of hate speech is implicit from the source of the dataset in the experiments. Our techniques can be applied to different definitions of hate speech and other applications as is described later.} and toxic content has been studied in various contexts, e.g., in reducing unintended bias, detecting covert hate speech, and mitigating hate speech on online platforms~\cite{kennedy_contextualizing_2020}. \cite{sachdeva_measuring_2022} have looked at how to quantify the intensity of hate speech and created labeled datasets. In~\cite{jain_controlling_2024}, the authors looked at controlling federated learning for hate speech classification. In this paper, we look at the problem of Bayesian agents identifying hate speech peddlers by sequentially parsing comments from users using an LLM.

\textit{Additional application: } In financial networks, LLMs can be used as sensors to parse textual information, including news articles and financial reports, and decisions are then taken based on the low-dimensional observations from the LLMs. Such a sequence of agents can herd in their decisions, leading to a financial bubble~\cite{chamley_rational_2004}.

\subsection{Organization}
Section~\ref{sec:formulation} formalizes the social learning framework and derives the social learning protocol for \bayesianagents\ sequentially parsing text observations to take actions minimizing the expected cost. Section~\ref{sec:herding} puts forth our main result, proving that herding occurs in the described social learning protocol. Section~\ref{sec:stochasticcontrol} proposes an online stochastic control approach for mitigating herding in \bayesianagents. Section~\ref{sec:application} describes an application of the proposed framework in identifying hate speech peddlers on social platforms and demonstrates numerical results on real-life hate speech datasets for the proposed approach. Section~\ref{sec:conclusion} concludes the paper, and the proofs for the main result are in the Appendix.

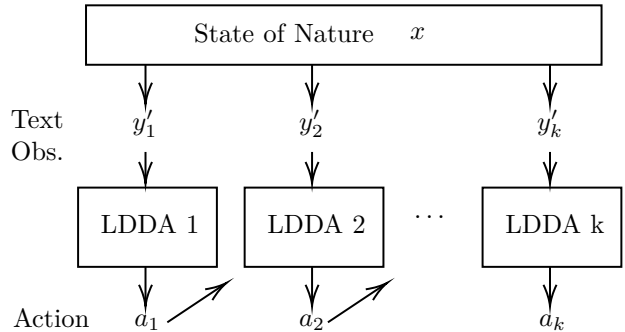
\begin{figure}[t!]
    \centering
    \vspace{3mm}
    \tikzset{every picture/.style={line width=0.75pt}} %set default line width to 0.75pt        

\begin{tikzpicture}[x=0.75pt,y=0.75pt,yscale=-1,xscale=1]
%uncomment if require: \path (0,485); %set diagram left start at 0, and has height of 485

%Straight Lines [id:da25103495615817883] 
\draw    (190,40) -- (190,58) ;
\draw [shift={(190,60)}, rotate = 270] [color={rgb, 255:red, 0; green, 0; blue, 0 }  ][line width=0.75]    (10.93,-3.29) .. controls (6.95,-1.4) and (3.31,-0.3) .. (0,0) .. controls (3.31,0.3) and (6.95,1.4) .. (10.93,3.29)   ;
%Straight Lines [id:da7299999444732306] 
\draw    (190,84) -- (190,98) ;
\draw [shift={(190,100)}, rotate = 270] [color={rgb, 255:red, 0; green, 0; blue, 0 }  ][line width=0.75]    (10.93,-3.29) .. controls (6.95,-1.4) and (3.31,-0.3) .. (0,0) .. controls (3.31,0.3) and (6.95,1.4) .. (10.93,3.29)   ;
%Shape: Rectangle [id:dp22431346951353182] 
\draw   (156,102) -- (226,102) -- (226,142) -- (156,142) -- cycle ;
%Straight Lines [id:da08664054686153011] 
\draw    (190,142) -- (190,160) ;
\draw [shift={(190,162)}, rotate = 270] [color={rgb, 255:red, 0; green, 0; blue, 0 }  ][line width=0.75]    (10.93,-3.29) .. controls (6.95,-1.4) and (3.31,-0.3) .. (0,0) .. controls (3.31,0.3) and (6.95,1.4) .. (10.93,3.29)   ;
%Shape: Rectangle [id:dp5132467356283734] 
\draw   (160,10) -- (420,10) -- (420,40) -- (160,40) -- cycle ;
%Straight Lines [id:da7645445274889251] 
\draw    (274,40) -- (274,58) ;
\draw [shift={(274,60)}, rotate = 270] [color={rgb, 255:red, 0; green, 0; blue, 0 }  ][line width=0.75]    (10.93,-3.29) .. controls (6.95,-1.4) and (3.31,-0.3) .. (0,0) .. controls (3.31,0.3) and (6.95,1.4) .. (10.93,3.29)   ;
%Straight Lines [id:da8820615564836409] 
\draw    (274,84) -- (274,98) ;
\draw [shift={(274,100)}, rotate = 270] [color={rgb, 255:red, 0; green, 0; blue, 0 }  ][line width=0.75]    (10.93,-3.29) .. controls (6.95,-1.4) and (3.31,-0.3) .. (0,0) .. controls (3.31,0.3) and (6.95,1.4) .. (10.93,3.29)   ;
%Shape: Rectangle [id:dp06976285002760485] 
\draw   (240,102) -- (310,102) -- (310,142) -- (240,142) -- cycle ;
%Straight Lines [id:da528615411032872] 
\draw    (274,142) -- (274,160) ;
\draw [shift={(274,162)}, rotate = 270] [color={rgb, 255:red, 0; green, 0; blue, 0 }  ][line width=0.75]    (10.93,-3.29) .. controls (6.95,-1.4) and (3.31,-0.3) .. (0,0) .. controls (3.31,0.3) and (6.95,1.4) .. (10.93,3.29)   ;
%Straight Lines [id:da15285071063020816] 
\draw    (394,40) -- (394,58) ;
\draw [shift={(394,60)}, rotate = 270] [color={rgb, 255:red, 0; green, 0; blue, 0 }  ][line width=0.75]    (10.93,-3.29) .. controls (6.95,-1.4) and (3.31,-0.3) .. (0,0) .. controls (3.31,0.3) and (6.95,1.4) .. (10.93,3.29)   ;
%Straight Lines [id:da9248345255396879] 
\draw    (394,84) -- (394,98) ;
\draw [shift={(394,100)}, rotate = 270] [color={rgb, 255:red, 0; green, 0; blue, 0 }  ][line width=0.75]    (10.93,-3.29) .. controls (6.95,-1.4) and (3.31,-0.3) .. (0,0) .. controls (3.31,0.3) and (6.95,1.4) .. (10.93,3.29)   ;
%Shape: Rectangle [id:dp39624877952183657] 
\draw   (360,102) -- (430,102) -- (430,142) -- (360,142) -- cycle ;
%Straight Lines [id:da4145777093089915] 
\draw    (394,142) -- (394,160) ;
\draw [shift={(394,162)}, rotate = 270] [color={rgb, 255:red, 0; green, 0; blue, 0 }  ][line width=0.75]    (10.93,-3.29) .. controls (6.95,-1.4) and (3.31,-0.3) .. (0,0) .. controls (3.31,0.3) and (6.95,1.4) .. (10.93,3.29)   ;
%Straight Lines [id:da4602948997720222] 
\draw    (201,170) -- (229.34,151.11) ;
\draw [shift={(231,150)}, rotate = 146.31] [color={rgb, 255:red, 0; green, 0; blue, 0 }  ][line width=0.75]    (10.93,-3.29) .. controls (6.95,-1.4) and (3.31,-0.3) .. (0,0) .. controls (3.31,0.3) and (6.95,1.4) .. (10.93,3.29)   ;
%Straight Lines [id:da4866744742659823] 
\draw    (282,170) -- (310.34,151.11) ;
\draw [shift={(312,150)}, rotate = 146.31] [color={rgb, 255:red, 0; green, 0; blue, 0 }  ][line width=0.75]    (10.93,-3.29) .. controls (6.95,-1.4) and (3.31,-0.3) .. (0,0) .. controls (3.31,0.3) and (6.95,1.4) .. (10.93,3.29)   ;

% Text Node
\draw (322,20) node [anchor=north west][inner sep=0.75pt]    {$\statevar$};
% Text Node
\draw (183,164.4) node [anchor=north west][inner sep=0.75pt]    {$a_{1}$};
% Text Node
\draw (182,61.4) node [anchor=north west][inner sep=0.75pt]    {$y_{1}^{\prime}$};
% Text Node
\draw (166,114) node [anchor=north west][inner sep=0.75pt]   [align=left] {LDDA 1};
% Text Node
\draw (213,18) node [anchor=north west][inner sep=0.75pt]   [align=left] {State of Nature};
% Text Node
\draw (135.5,75.5) node   [align=left] {Text \\Obs.};
% Text Node
\draw (266,164.4) node [anchor=north west][inner sep=0.75pt]    {$a_{2}$};
% Text Node
\draw (266,61.4) node [anchor=north west][inner sep=0.75pt]    {$y_{2}^{\prime}$};
% Text Node
\draw (250,114) node [anchor=north west][inner sep=0.75pt]   [align=left] {LDDA 2};
% Text Node
\draw (387,164.4) node [anchor=north west][inner sep=0.75pt]    {$a_{k}$};
% Text Node
\draw (386,61.4) node [anchor=north west][inner sep=0.75pt]    {$y_{k}^{\prime}$};
% Text Node
\draw (370,114) node [anchor=north west][inner sep=0.75pt]   [align=left] {LDDA k};
% Text Node
\draw (324,114.4) node [anchor=north west][inner sep=0.75pt]    {$\dotsc $};
% Text Node
\draw (142.5,168) node   [align=left] {\begin{minipage}[lt]{31.07pt}\setlength\topsep{0pt}
\begin{center}
Action
\end{center}

\end{minipage}};

\end{tikzpicture}
    \vspace{-10mm}
    \caption{Schematic representation of social learning with \bayesianagentstext\ (\bayesianagent) with LLM sensors for compressing the observations, ($\llmobservation_\alttimeindex$) and estimating the underlying states, ($\statevar_\alttimeindex$). The \bayesianagent\ updates their prior over the state space using past actions ($\action_\alttimeindex$) and likelihood (parameterized by a neural network) for the compressed observations.}
    \label{fig:multipleagent}
    \vspace{-7mm}
\end{figure}
\section{Social Learning for \bayesianagentstext\ (\bayesianagents) }\label{sec:formulation}
This section formalizes the Bayesian social learning setup involving \bayesianagents. We discuss the system model and assumptions on the state, observation, and action space. In Section~\ref{sec:application}, we map out each component of the framework and motivate with respect to the application of hate speech peddler detection on online platforms. The framework is illustrated for single and multiple agents in Fig.~\ref{fig:singleagent} and Fig.~\ref{fig:multipleagent}, respectively.

The discrete state space is denoted by $\statespace$ and has cardinality $\statedim$. $\probspace(\statedim)$ denotes the probability simplex over the belief space. In the standard social learning literature, the state of nature evolves as a Markov chain with transition probability kernel $\markovkernel$. However, in the following sections, we concern ourselves with estimating the state of nature chosen at the start, $\statevar \in \statespace$, i.e., $\markovkernel=I$. A noisy observation, $\llmobservation_\timeindex \in \textobservationspace$  of state $\statevar$ is observed with probability $\probabilitymeasure(\llmobservation_\timeindex|\statevar)$, where $\textobservationspace$ is a discrete high-dimensional space.

A sequence of \bayesianagentstext\ (\bayesianagents) aim to estimate the state $\statevar$ given the text observations ($\llmobservation_\timeindex$) and use a large language model (LLM) as a sensor probe to achieve this. The LLM probe is a noisy map from the high-dimensional space $\textobservationspace$ to a low-dimensional discrete feature space $\observationspace$. For a text observation $\llmobservation \in \textobservationspace$, the LLM gives a low-dimensional observation $\observation\in\observationspace$ with probability $\probabilitymeasure(\observation_\timeindex|\llmobservation_\timeindex)$. Hence for state $\statevar$, the \bayesianagent\ recieves a noisy observation $\observation_\timeindex \in \observationspace$ with a probability $\probabilitymeasure(\observation_\timeindex|\statevar) = \sum_{\llmobservation}\probabilitymeasure(\observation_\timeindex|\llmobservation_\timeindex)\probabilitymeasure(\llmobservation_\timeindex|\statevar)$. $\observationmatrix \in \R^{\statedim\times|\observationspace|}$ is the observation matrix.

These noisy observations are fed into a \bayesianagent\ which acts as a detector of the state and computes the posterior,
    \begin{align}\label{eq:bayesrule}
\probabilitymeasure(\statevar|\observation_\timeindex) = \frac{\probabilitymeasure(\observation_\timeindex|\statevar)\prior_\timeindex(\statevar)}{\sum_{\statevar^\prime \in \statespace} \probabilitymeasure(\observation_\timeindex|\statevar^\prime)\prior_\timeindex(\statevar^\prime)},
    \end{align}
    where $\prior_\timeindex$ is the prior at time $\timeindex$. Based on the posterior and a cost function $\cost:\statespace\times\actionspace\to\R$, the \bayesianagent\ outputs  an action $\action \in \actionspace$ minimizing an expected cost with respect to the posterior probability,
    \begin{align}\label{eq:action}
        \action_\timeindex = \underset{\action \in \actionspace}{\argmin}{\sum_{\statevar}\cost(\statevar,
        \action)\probabilitymeasure(\statevar|\observation_\timeindex)}.
    \end{align}
    
\begin{algorithm}[h!]
    \begin{algorithmic}[1]
    \State  Agents aim to estimate state $\statevar$
    \For{$\timeindex \in 1,2,\dots$}
    \State Agent $\timeindex$ observes $\llmobservation_\timeindex\sim \probabilitymeasure(\llmobservation_\timeindex|\statevar)$ 
    \State Agent $\timeindex$ uses LLM to obtain $\observation_\timeindex\sim \probabilitymeasure(\observation_\timeindex|\llmobservation_\timeindex)$ 
    \State Agent $\timeindex$ computes posterior using~\eqref{eq:bayesrule} 
    \State Agent $\timeindex$ takes optimal action according to~\eqref{eq:action}
    \State Agents $\timeindex+1,\dots$  update public prior using~\eqref{eq:priorupdate}
    \EndFor{}
    \end{algorithmic}
    \caption{Social Learning Protocol for \bayesianagents}
    \label{alg:sociallearning}
\end{algorithm}

The \bayesianagent\ broadcasts its action to all successive \bayesianagents. The future \bayesianagents\ update their prior of the true state based on observing the action of the \bayesianagent\ (and not the private observation), 
    \begin{align}\label{eq:priorupdate}
        \prior_{\timeindex + 1} = \priorupdate(\prior_{\timeindex},\action_\timeindex).
    \end{align}
Here, the $\priorupdate$ is given by the following filtering equation, 
\begin{align}
 \priorupdate(\prior,\action) = \frac{\actionprob(\prior,\action) \prior}{\indicator_{\statedim}\actionprob(\prior,\action)  \prior} ,
\end{align}
where $\actionprob(\prior,\action) = \mathsf{diag}(\probabilitymeasure(\action|\statevar=1,\prior),\dots,\probabilitymeasure(\action|\statevar=\statedim,\prior))$ is the probability of actions for different states given the prior, each of which is given by, 
\begin{align}\label{eq:probaction}
    \begin{split}
\probabilitymeasure(\action|\statevar=\stateidx,\prior) = \sum_{\observation\in \observationspace} \probabilitymeasure(\action|\observation,\prior)\probabilitymeasure(\observation|\statevar=\stateidx,\prior)\\ \probabilitymeasure(\action|\observation,\prior) 
 = \begin{cases}
     1, & \text{if } \cost^\prime_\action \observationmatrix_\observation  \prior \leq \cost^\prime_{\actionalt} \observationmatrix_\observation  \prior, \actionalt \in \actionspace \\
     0, &\text{otherwise}
 \end{cases}\\
    \end{split}
\end{align}
where $\observationmatrix_\observation = \mathsf{diag}([\probabilitymeasure(\observation|\statevar=1),\dots, \probabilitymeasure(\observation|\statevar=\statedim)])^\prime$ and $\cost_\action = [\cost(1,\action),\dots,\cost(\statedim,\action)]^\prime$. The derivation of the social learning filter can be found in~\cite{krishnamurthy_partially_2016}.  The next section shows that ~\bayesianagents\ herd in their public belief in finite time with probability 1. 

\section{Main Result: Herding in Social Learning}\label{sec:herding}
This section proves that the \bayesianagents\ described in the previous section form an information cascade and herd in their actions when the public prior gets strong. 

We first define an information cascade occurring at time $\timeindex$ for the \bayesianagents\ in the following definition.

\begin{definition}\label{def:infocascade}(\textbf{Information Cascade}): 
    An information cascade occurs at time $\herdtime$ if the public belief of all agents after time $\herdtime$ are identical, i.e., $\prior_{\timeindex}(\statevar) = \prior_{\herdtime}(\statevar)$ $\forall \ \statevar\in\statespace$ for all time $\ \timeindex\geq\herdtime$. 
\end{definition}
Next, we define herding at time $\herdtime$ for \bayesianagents.
\begin{definition}\label{def:herding}(\textbf{Herding}):
    Herding in the \bayesianagents\ agents takes place at time $\herdtime$ if the action of all agents after $\herdtime$ are identical, i.e. $\action_{\timeindex} = \action_{\herdtime}$ for all time $\timeindex\geq\herdtime$. 
\end{definition}
It is straightforward to show that an information cascade (Def.~\ref{def:infocascade}) occurring at time $\timeindex$ implies that herding also takes place at time $\timeindex$ (Def.~\ref{def:herding}). We now state the main result on herding in \bayesianagents, which shows that the protocol of Algorithm~\ref{alg:sociallearning} leads to the agents herding in finite time. 
\begin{theorem}\label{th:herding}(\textbf{Herding in Bayesian social learning of \bayesianagents})
    The social learning protocol of the \bayesianagents\ described in Algorithm~\ref{alg:sociallearning} leads to an information cascade (Def.~\ref{def:infocascade}) and agents herd (Def.~\ref{def:herding}) in finite time $\herdtime<\infty$ with probability 1.
\end{theorem}
\begin{proof}
    Proof in Appendix.
\end{proof}
Theorem~\ref{th:herding} shows that herding happens in finite time, and therefore, the agents take the same action regardless of their private observation. Discarding the private observation, which provides valuable information about the current state, makes their state estimation incorrect and inefficient.    

A stochastic control approach can be deployed for the \bayesianagents\ that detects herding and counteracts it by suitably controlling the actions based on the observations and the public prior. 
However, to derive structural results on the optimal policy of the controller, we formulate a stopping time problem that optimally decides when to herd based on a socialistic objective. 
\section{Stochastic Control Approach for Delayed Herding in \bayesianagents}\label{sec:stochasticcontrol}
The herding can be controlled dynamically by switching between the optimal action computed using~\eqref{eq:action} and revealing the private observation. However, we consider a stopping time problem formulation for quickest-time herding to derive structural results on the optimal decision rule. The stopping time formulation makes the \bayesianagents\ act benevolently to delay herding, make observations public and improve the state estimate. 
 
\subsection{Socialistic Objective for Stopping Time Problem}
We formulate a stopping time problem to switch between two modes of action optimally, one with complete privacy (acting according to~\eqref{eq:action}) or with no privacy (revealing private observation). The decision rule for the type of action (reveal private observation or optimally act) is denoted by $\decision$ and the stationary rule for computing the decision rule by $\policy$. 
To formulate the stopping time problem for quickest time herding in \bayesianagents, we first define the following natural filtration at time $\timeindex$ with respect to the actions and decision rules of the past agents,
\begin{align*}
    \filtration_\timeindex = \sigma(\action_1,\dots,\action_{\timeindex-1},\decision_1,\dots,\decision_\timeindex).
\end{align*}
The objective of the stopping time is to optimally achieve a tradeoff between delaying herding and announcing the state $\indicatorstate_1$ (by herding to it). We denote $\stoppingtime$ as the stopping time when the decision to stop is announced (after which the agents herd). 
Then we define the following social welfare cost, which each \bayesianagent\ optimizes with respect to the stationary policy $\policy$, 
\begin{align}
\begin{split}\label{eq:socialwelfarecost}
    &\avgcost_{\policy}(\prior) = \expectation_{\policy}\left\{  
    \sum_{\timeindex = 1}^{\stoppingtime-1} \discountfactor^{\timeindex-1} \expectation\left\{\cost(\statevar_\timeindex,\action_\timeindex)|\filtration_{\timeindex-1}\right\} \right.\\&\left.+ \sum_{\timeindex=1}^{\stoppingtime-1} \discountfactor^{\timeindex-1} \delaycost  \expectation_{\policy} \left\{\indicator(\statevar_\timeindex=\indicatorstate_1)|\filtration_{\timeindex-1}
    \right\} +\right.\\&\left. 
    \discountfactor^{\stoppingtime-1} \errorcost\expectation_\policy\left\{ \indicator(\statevar_{\stoppingtime}\neq \indicatorstate_1 )\right\} + \frac{\discountfactor^{\stoppingtime-1}}{1-\discountfactor} \min_{\action \in \actionspace} \expectation\left\{ \cost(\statevar_\stoppingtime,\action)|\filtration_{\stoppingtime-1}\right\}
    \right\},
\end{split}
\end{align}
here $\delaycost$ and $\errorcost$ are cost parameters and $\discountfactor$ is the discount factor. 

\textit{Justification for cost: } The first term in the cost represents the discounted cost incurred before the decision to stop is announced. The second and third terms correspond to the delay in announcing the decision to stop and the error in announcing the stopping. The final term corresponds to the infinite horizon cost incurred once the stopping has been announced.  Hence, the cost of~\eqref{eq:socialwelfarecost} is a socialistic cost that accounts for a discounted cost that'll be incurred instead of the opportunistic cost-based action of~\eqref{eq:action}.

The social welfare cost of~\eqref{eq:socialwelfarecost} is minimized by using a decision rule which decides the action the \bayesianagent\ takes, 
\begin{align}\label{eq:constraineddecisionrule}
\action_\timeindex(\prior_{\timeindex-1},\observation_\timeindex,\policy) = \begin{cases}
        \observation_\timeindex, & \text{if } \policy(\prior_{\timeindex-1}) = 2\\
        \argmin_\action \cost_\action \prior_{\timeindex-1}& \text{if } \policy(\prior_{\timeindex-1}) = 1 
    \end{cases},
\end{align}
where the policy $\policy:\probspace(\statedim)\to \decisionspace$ maps the belief over the state space to the decision space $\decisionspace = \{1, 2\}$, and the decision $1$ corresponds to continuing whereas $2$ corresponds to stopping. In the following subsection, we show that the optimal policy for the social welfare cost of~\eqref{eq:socialwelfarecost}, has a switching threshold curve with respect to the belief space. 
\subsection{Structural Results}
To prove our structural result, we make the following assumptions on our cost structure, 
\begin{enumerate}[start=1,label={(\bfseries S\arabic*):}]
    \item $\cost(\indicatorstate_\stateidx,\action) - \cost(\indicatorstate_{\stateidx+1},\action) \geq 0$ $\forall \stateidx = 1,\dots,\statedim-1$ \ $\forall \action$.
    \item $\cost(\indicatorstate_\statedim,\action) - \cost(\indicatorstate_\stateidx,\action) \geq (1-\discountfactor)\sum_{\observation}(\cost(\indicatorstate_\statedim,\action)\observationmatrix_{\statedim,\observation} - \cost(\indicatorstate_\stateidx,\action)\observationmatrix_{\stateidx,\observation})$ $\forall \stateidx = 1,\dots,\statedim$.
    \item $(1-\discountfactor)\sum_{\observation}(\cost(\indicatorstate_1\action)\observationmatrix_{1,\observation} - \cost(\indicatorstate_\stateidx,\action)\observationmatrix_{\stateidx,\observation}) \geq \cost(\indicatorstate_1,\action) - \cost(\indicatorstate_\stateidx,\action)$ $\forall \stateidx = 1,\dots,\statedim$.
    \item $\observationmatrix$ is totally positive of order 2 (TP2)~\footnote{A matrix $A$ is TP2 if all minors of the matrix $A$ are positive.}.
\end{enumerate}
\textit{Justification for Assumptions: } S1 is sufficient for the cost terms to be monotone likelihood ratio decreasing (defined in Appendix~\ref{sec:expapp}) in $\prior$~\cite{krishnamurthy_partially_2016}. This condition can be ensured by taking a cost that puts the state $1$ as the most costly. S2 and S3 are sufficient for the costs to be submodular, and they imply that the difference between the cost of continuing and the cost of stopping decreases in the state space. This gives incentives to herd when the belief is strong enough towards $\indicatorstate_1$. Intuitively, the decision to herd should be made when making the observations public, which does not improve the state estimate. S4 is a standard structural assumption in analyzing structural results for optimal policies of stopping time partially observed Markov decision processes.
\begin{theorem}\label{th:threshold}
    Consider the sequential decision problem of \bayesianagents\ for detecting state $\indicatorstate_1$ with the social welfare cost of \eqref{eq:socialwelfarecost} and the constrained decision rule of \eqref{eq:constraineddecisionrule}. Then, under Assumption (S1-4), the constrained decision rule of \eqref{eq:constraineddecisionrule} is a threshold in the belief space with respect to a threshold switching curve that partitions the belief space $\probspace(\statedim)$. The optimal policy can be given by, 
    \begin{align}\label{eq:optpolicy}
        \policy^{*}(\prior) = \begin{cases}
           2 \ \text{(continue) if} \ \prior \in \region_2 \\
           1 \ \text{(stop) if} \ \prior \in \region_1 \\
        \end{cases},
    \end{align}
    where $\region_1$ and $\region_2$ are individual connected regions of $\probspace(\statedim)$.
\end{theorem}
\begin{proof}
    Proof in Appendix. 
\end{proof}
The optimal policy being threshold along a switching curve makes it possible to compute the optimal policy efficiently. A linear threshold function can approximate the switching curve. The optimal linear approximation can be searched by using the approximate estimates of the expected cost of~\eqref{eq:socialwelfarecost}, e.g., by using a simultaneous perturbation stochastic approximation algorithm which adaptively updates the parameters of the linear threshold to minimize the cost.

\section{Key Application and Numerical Results: Flagging Hate Speech Peddlers on Online Platforms}\label{sec:application}
This section discusses the key proposed application of our framework in flagging hate speech peddlers on a social network. We motivate the application and describe the different components of our framework. A numerical experiment is shown on a real-life hate speech dataset to show that a sequence of \bayesianagents\ forms an information cascade. Finally, we numerically show how a threshold policy delays herding in the stopping time problem formulation of Section~\ref{sec:stochasticcontrol}. 

\subsection{\bayesianagents\ detecting hate speech peddlers on online platforms: Motivation and Experimental Setup}
Flagging users who spread toxic content online is a significant challenge. The state $\statevar_\timeindex$ could represent the category of peddlers classified based on the intensity and type of content they are propagating. 

For example, the state could be 3-dimensional, indicating the hate-intent of the user (hateful or not), the hate speech intensity scale~\cite{bahador_classifying_nodate}, and the particular group the hate speech is directed towards. The noisy observations are the text comments from the user that inform about the state and are from a high-dimensional observation space.  The \bayesianagent\ is an online detection mechanism equipped with a large language model (LLM) sensor that analyses comments and flags users sequentially. The LLM is used to parse the text and obtain a list of appropriate features from the finite-dimensional feature space $\observationspace$, which the platform could design. The features contain information about the text and comparisons of the text with a given context. In our setup, the Bayesian engine of the \bayesianagent\ consists of a likelihood parameterized by a neural network. For a discrete distribution, the likelihood neural network could use a restricted Boltzmann machine (RBM) to generate samples from the likelihood. The posterior can be updated using~\eqref{eq:bayesrule} using the likelihood and the closed form prior of~\eqref{eq:priorupdate} from the previous step. The \bayesianagent\ takes an action according to~\eqref{eq:action}. The action is classifying whether the user has hateful intent, and the cost accounts for the misclassification of the state. 

In a social network, there may be multiple such \bayesianagent\ deployed to flag malicious users and decrease propagation of hate speech. The flags by the previous \bayesianagent\ are visible, but the private observations are not due to computation and privacy restrictions (so that the LLM can not be fine-trained on the text or the feature mappings). Since the observations are generated sequentially, the \bayesianagent\ may use the same LLM but with a different context and for a different text observation. Therefore, a single \bayesianagent\ can be viewed as a sequence of \bayesianagents\ learning from their private observations and the past actions of previous \bayesianagents.

\subsection{Numerical Demonstration of Herding}\label{sec:numerical}
\begin{figure}
    \centering
\includegraphics[width=\columnwidth]{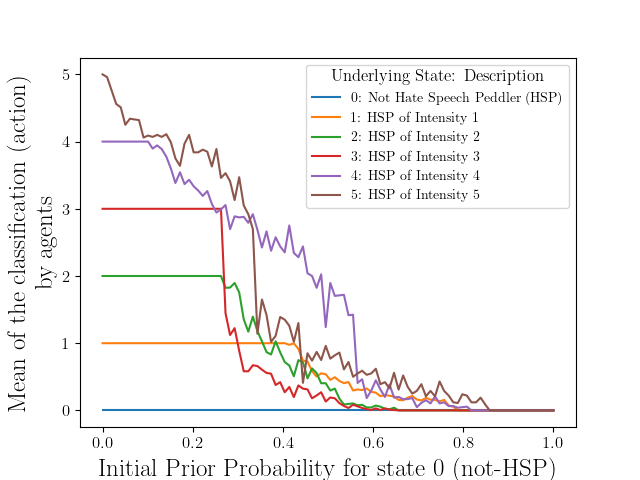}
    \caption{Herding in \bayesianagents. For different underlying states, the average action of the agents is $0$ for a strong initial prior probability. Hence if the public prior is strong, the \bayesianagents\ misclassify the user as a non hate speech peddler even in presence of private evidence otherwise. }
    \label{fig:herding}
    \vspace{-5mm}
\end{figure}
We now numerically demonstrate the existence of information cascades when \bayesianagents\ sequentially perform state estimation for detecting hate speech peddlers. 

%%% State, Observation, Action Space, Cost Structure, Likelihood
The state space models the state of the user, $\statespace = \{0= \text{non-hateful}, 1 = \text{hateful}\} \times \{1,2,3,4,5\}$, where the first dimension corresponds to whether a user is a hate speech peddler (HSP) or not, and the second dimension to the intensity of the toxicity of the speech (evaluated by crowdsourcing~\cite{sachdeva_measuring_2022}). For the numerical result, we consider an augmented state space with $6$-states, $\statespace=\{0=\text{non-hateful},1,2,3,4,5\}$, where the last $5$ entries denote an HSP of different intensities as described later. The high-dimensional observations of the state are in the form of text comments posted on online platforms. An LLM is used to parse the text observations by prompting the LLM with the text and a system prompt to return an output belonging to an observation space $\observationspace$, which contains the following binary variables: a) targetted towards someone and b) contains explicit words c) indicate violence d) has bias e) is dehumanizing f) is genocidal. We augment the output of the LLM to an observation space of cardinality $|\observationspace| = 6$. The details of this augmentation, along with additional experimental details, are in Appendix~\ref{sec:expapp}.  We use Mixtral 8x7B v0.1, an open-source mixture of experts LLM with 7 billion parameters~\cite{jiang_mixtral_2024}. We query the LLM using the TogetherAI API.  The reproducible code and the dataset link can be found on github.com/aditj/sociallearningllm. 

The action space $\actionspace$ is considered the same as the state space and the cost function which accounts for the misclassification of an HSP is given by,
\begin{align*}
    \cost(\statevar,\action) &= \indicator(\statevar\neq 0)[\indicator(\action=0)+ |\statevar-\action|],
\end{align*}
where square brackets are the components of the vector. The first time accounts for the misclassification of a hateful user and the second term accounts for the difference in intensity.  
We use the measuring hate speech dataset from~\cite{sachdeva_measuring_2022}, which contains $40000$ annotated comments from Twitter, Reddit, and Gab. The annotations are performed by crowdsourcing and indicate if the comments contain hate speech and the intensity of the toxicity exhibited on a scale of 1 to 5, measured using a Rasch measurement scale whose details can be found in~\cite{sachdeva_measuring_2022}. 

Since the data is anonymized, we consider a synthetic user construction. In a span of $\timehorizon$ textual comments, a hate speech peddler (HSP) contains hate speech text from one of the intensity levels. Hence there are $6$ types of users: non-HSP and HSP with intensity from 1 to 5, each with $\timehorizon=100$ comments of the corresponding intensity. The experiment results are averaged over $N_{MC}=100$ independent runs. The likelihood is computed by training a restricted Boltzmann machine, a subset of the training dataset with $1000$ samples, the details of which are in the Appendix. 

For both experiments, we benchmark our results against the initial public belief (specifically prior probability of state 0). Due to Theorem~\ref{th:herding}, the initial public belief is sufficient to identify intermediate public belief regions where the information cascades are observed.

%%% Results
In Figure~\ref{fig:herding}, we report the average action for the agents for different initial prior beliefs and initial states. We only vary the initial prior belief for state $0$ and assume the rest of the initial prior over the actual underlying state. In regions of initial solid strong prior, the average action is $0$ regardless of the true state. The difference in the public belief region for information cascades depends on the observation probabilities.

\subsection{Threshold policy for delayed herding in stopping time problem}
\begin{figure}[th!]
\vspace{-5mm}
    \centering
\includegraphics[width=\columnwidth]{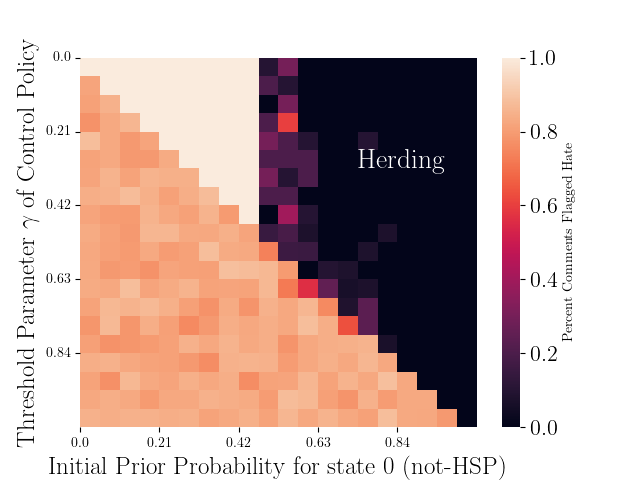}
    \caption{Effect of different thresholds in the policy of form~\eqref{eq:thresholdpolicy}. For the state $\statevar=1$, it can be seen that increasing the policy threshold delays herding in regions of strong priors (left and right). The optimal threshold, which can be searched using stochastic approximation, is a threshold that optimally delays herding to balance the delay, error, and social learning costs of~\eqref{eq:socialwelfarecost}.  }
    \label{fig:thresholding}
\end{figure}  
\vspace{-2mm}
Since the optimal policy of~\eqref{eq:optpolicy} has a threshold structure, we demonstrate the effect of the policies with different threshold levels in delaying herding in the stopping time problem of Sec.~\ref{sec:stochasticcontrol}. The optimal threshold policy can be searched using a stochastic approximation algorithm, which will be discussed in a later work. We consider a simplified state space $\statespace = \{0 = \text{not hateful},1=\text{hateful}\}$ and consider the task of estimating the state $\indicatorstate_1 = 1 =$ hateful. Correspondingly, we consider a reduced action space $\actionspace = \{0 = \text{not hateful},1=\text{hateful}\}$ and a reduced observation space which gives a binary observation as to whether the text is toxic or not. We consider the cost $\cost(\statevar,\action) = \indicator(\statevar=1)\indicator(\action=0)$, maximizing the type-1 error for $1$ state. We assume that the observation matrix is such that an HSP gives a text observation that is toxic with probability $70\%$, whereas a non-HSP does not give a toxic text observation. The other parameters remain the same from~\ref{sec:numerical}. For the simplified system model, policies that have a switching threshold curve with respect to the belief space can be represented as, 

 \begin{align}\label{eq:thresholdpolicy}
        \policy(\prior) = \begin{cases}
           2 \ \text{(continue) if} \ \prior(0) \leq \threshold \\
           1 \ \text{(stop) if} \ \prior(0) > \threshold  \\
        \end{cases},
    \end{align}
where $\threshold$ is the threshold parameter. 

We run the experiment when the underlying state $\statevar=1$ and measure the percentage of actions when the user is not flagged as an HSP. 
In Figure~\ref{fig:thresholding}, we demonstrate the effect of having different thresholds for different initial public beliefs in delaying herding even when the initial public belief is strong for the state $0$ (not an HSP). Ideally, the percentage should be ~$70\%$ region of Figure~\ref{fig:thresholding} (lower-triangular region). We observe that in regions of strong public belief for state $0$ (right side), a decrease in the threshold (more private observations revealed) delays herding and takes actions closer to the true state (lower right corner). The optimal threshold will optimally delay in herding to minimize the cost of~\eqref{eq:socialwelfarecost}.

\section{Conclusion}\label{sec:conclusion}
This paper considered the social learning problem in \bayesianagenttext\ using large language models as a sensor to parse textual observations from the environment. It was shown that a sequence of \bayesianagents\ herd to a common action when the private observations are hidden. A stopping time problem was formulated to optimally tradeoff between privacy and herding, the optimal policy with a threshold structure. We numerically illustrated the herding and effect of a threshold policy on a real-life hate speech dataset where the \bayesianagents\ aim to identify hate speech peddlers in an online platform. 
The single \bayesianagent\ can be of independent research interest to synthesize LLM aided systems making robust autonomous decisions.
Future work can look at the problem of optimally incentivizing a sequence of LLM agents in a stopping-time formulation to delay herding and estimate the underlying state optimally. 
\bibliographystyle{abbrv}
\bibliography{references}

\appendix\label{sec:appendix}
    \subsection{Proof for Theorem~\ref{th:herding}}
    \begin{proof}
    Define $\priorratio_\timeindex(\stateidx,\stateidxalt) = \log(\prior(\stateidx)/\prior(\stateidxalt)), \stateidx,\stateidxalt \in \statespace$. From~\eqref{eq:priorupdate} we have, $\priorratio_{\timeindex+1}(\stateidx,\stateidxalt) = \priorratio_{\timeindex}(\stateidx,\stateidxalt) + \actionratio_{\timeindex}(\stateidx,\stateidxalt)$ where $\actionratio_{\timeindex}(\stateidx,\stateidxalt) = \log(\probabilitymeasure(\action_\timeindex|\statevar=\stateidx,\prior_\timeindex)/\probabilitymeasure(\action_\timeindex|\statevar=\stateidxalt,\prior_\timeindex)$. 

    The probability of the actions given the state and prior can be written as, 
    \begin{align*}
        \probabilitymeasure(\action|\statevar,\prior) = \sum_{\observation\in \observationspace} \prod_{\actionalt \in \actionspace \setminus \action} \indicator(\cost^\prime_\action \observationmatrix_\observation  \prior \leq \cost^\prime_{\actionalt} \observationmatrix_\observation  \prior) \observationmatrix_{\observation,\statevar}.
    \end{align*}
    Let $\Tilde{\observationspace}_\timeindex \subseteq \observationspace$ be a subset of the observation space for which the action $\action_\timeindex$ is suboptimal with respect to all other actions, i.e., $\prod_{\actionalt \in \actionspace \setminus \action_\timeindex}  \indicator(\cost^\prime_{\action_\timeindex} \observationmatrix_\observation  \prior > \cost^\prime_{\actionalt} \observationmatrix_\observation  \prior) = 1 \ \forall \observation \in \observationspace$. When information cascade (Def.~\ref{def:infocascade}) occurs this set should be empty since no matter what the observation, the action $\timeindex$ should be optimal according to~\eqref{eq:action}.
    Also, rewriting $\actionratio_{\timeindex}(\stateidx,\stateidxalt)$,
    \begin{align*}
        \actionratio_{\timeindex}(\stateidx,\stateidxalt) = \log \left(\frac{1-\sum_{\observation\in\Tilde{\observationspace}_\timeindex}\observationmatrix_{\observation,\stateidx}}{1-\sum_{\observation\in\Tilde{\observationspace_\timeindex}}\observationmatrix_{\observation,\stateidxalt}}\right).
    \end{align*}
    Therefore when an information cascade occurs, $\actionratio_\timeindex(\stateidx,\stateidxalt) = 0, \ \forall \stateidx,\stateidxalt\in\statespace$ (Due to $\Tilde{\observationspace_\timeindex}$ being an empty set). Also if $\Tilde{\observationspace_\timeindex}$ is nonempty, then $\actionratio_\timeindex(\stateidx,\stateidxalt)>\constant$, where $\constant$ is a positive constant.   

    Let $\actionfiltration_\timeindex = \sigma(\action_1,\action_2,\dots,\action_\timeindex)$ denote the natural filtration, where $\sigma$ is the operator which generates the corresponding sigma field.

    $\prior_\timeindex(\stateidx) = \probabilitymeasure(\statevar=\stateidx|\action_1,\dots,\action_\timeindex) = \expectation[\indicator(\statevar = \stateidx)\vert\actionfiltration_\timeindex]$ is a martingale adapted to $\actionfiltration_\timeindex$ for all $\stateidx \in \statespace$. This follows by the application of smoothing property of conditional expectation, $\expectation[\prior_{\timeindex+1}(\stateidx)\mid \actionfiltration_\timeindex] = \expectation\left[\expectation[\indicator(\statevar = \stateidx)\mid\actionfiltration_{\timeindex+1}]\right] = \expectation[\indicator(\statevar=\stateidx)\mid\actionfiltration_\timeindex]$.

    Therefore by the  martingale convergence theorem, there exists a random variable $\prior_\infty$ such that, $
    \prior_\timeindex \to \prior_\infty \ \text{w.p.} 1$. Therefore $\priorratio_\timeindex(\stateidx,\stateidxalt) \to \priorratio_\infty(\stateidx,\stateidxalt)$ w.p. 1., which implies there exists $\Tilde{\timeindex}$ such that $\forall \timeindex \geq \Tilde{\timeindex}$, $\vert \priorratio_\infty(\stateidx,\stateidxalt) - \priorratio_\timeindex(\stateidx,\stateidxalt) \vert \leq \constant/3$ and so,
    \begin{align}\label{eq:contra1}
        \vert \priorratio_{\timeindex+1}(\stateidx,\stateidxalt) - \priorratio_\timeindex(\stateidx,\stateidxalt) \vert \leq 2\constant/3, \forall \timeindex \geq \Tilde{\timeindex}.
    \end{align}
    We now prove the theorem by contradiction. Suppose a cascade does not occur, then for at least one pair $\stateidx \neq \stateidxalt, \stateidx,\stateidxalt \in \statespace$, $\probabilitymeasure(\action|\statevar=\stateidx,\prior)$ is different than $\probabilitymeasure(\action|\statevar=\stateidxalt,\prior)$. This would imply that the set $\Tilde{\observationspace}_\timeindex$ is nonempty and therefore, 
    \begin{align}\label{eq:contra2}
        |\priorratio_{\timeindex}(\stateidx,\stateidxalt)| = \vert \priorratio_{\timeindex+1}(\stateidx,\stateidxalt) - \priorratio_\timeindex(\stateidx,\stateidxalt) \vert \geq \constant.
    \end{align}
    ~\eqref{eq:contra1} and \eqref{eq:contra2} contradict each other. Therefore $\probabilitymeasure(\action|\statevar=\stateidx,\prior)$ is same for all $\stateidx\in\statespace$ and hence according to~\eqref{eq:priorupdate} information cascade occurs at time $\Tilde{\timeindex}$.
    \end{proof}
    \subsection{Proof for Theorem~\ref{th:threshold}}
    \begin{proof}
    We prove the Theorem by showing that it satisfies the conditions of Theorem 12.3.4 of~\cite{krishnamurthy_partially_2016}. A more general proof can be found in~\cite{krishnamurthy_partially_2016,krishnamurthy_bayesian_2011-1}.

In order to state and verify the assumptions of Theorem 12.3.4 of~\cite{krishnamurthy_partially_2016}, we need to define first order stochastic dominance (FOSD) and a submodular function. 
    We first define a Monotone Likelihood Ratio (MLR) ordering on a line and then define a submodular function with respect to this MLR ordering. We only need to consider the following lines, 
    \begin{align*}
\lines(\indicatorstate_\stateidx,\bar{\prior}) = \{\prior \in \probspace(\statedim)  : \prior = (1-\epsilon)\bar{\prior} + \epsilon \indicatorstate_\stateidx, 0 \leq \epsilon\leq 1\}\\, \bar{\prior} \in \mathcal{H}_\stateidx,
    \end{align*}
    where the state index is only between the extreme states, $\stateidx \in \{1,\statedim \}$ and,
     \begin{align*}
        \mathcal{H}_\stateidx = \{\prior \in \probspace(\statedim):\policy(\stateidx) = 0\}.
    \end{align*}
    To define the MLR ordering on a line we first define the MLR ratio with respect to belief space, 
    \begin{definition}
       \textbf{(Monotone Likelihood Ratio (MLR) Order)} Let $\prior_1$,$\prior_2 \in \probspace(\statedim)$, then $\prior_1$ dominates  $\prior_2$ with respect to the MLR order ($\prior_1 \geq_r \prior_2$) if,
        \begin{align*}
\prior_1(\stateidx)\prior_2(\stateidxalt) \leq \prior_1(\stateidxalt)\prior_2(\stateidx), \ \ \stateidx < \stateidxalt, \stateidx,\stateidxalt \in \{1,\dots,\statedim\}.
        \end{align*}
    \end{definition}
    The following definition is for the MLR ordering the lines $\lines_{\indicatorstate_\stateidx}, \ \stateidx \in \{1,\statedim \}$ 
    \begin{definition}
        (\textbf{MLR Ordering on Line  $\geq_{\linesymb_\stateidx}$}) $\policy_1$ is greater than $\policy_2$ with respect to the MLR ordering on the line $\lines(\indicatorstate_\stateidx,\prior), \stateidx\in\{1,\statedim\}$ ($\prior_1 \geq_{\linesymb_\stateidx}\prior_2$), if $\prior_1,\prior_2 \in \lines(\indicatorstate_\stateidx,\bar{\prior})$ for some $\bar{\prior}$ and $\prior_1 \geq_r \prior_2$. 
    \end{definition}
    Finally we are ready to define a submodular function on a line,
    \begin{definition}
        \textbf{Submodular Function on Line} For $\stateidx\in\{1,\statedim\}$, a function
    $\phi: \lines(\indicatorstate_\stateidx,\bar{\prior})\times \decisionspace \to \R$ is submodular if $\phi(\prior,\decision) - \phi(\prior,\bar{\decision}) \leq \phi(\bar{\prior},\decision) - \phi(\bar{\prior},\bar{\decision})$, for $\bar{\decision}\leq \decision, \bar{\prior} \leq_{\linesymb_\stateidx} \prior$.
    \end{definition}
  % A function $\phi:\probspace\to\R$ is said to be MLR increasing if $\prior_1 \geq_r \prior_2$ implies $\phi(\prior_1) \geq \phi(\prior_2)$.

  The following is used extensively to compare two beliefs and is a weaker condition than MLR ordering, 
    \begin{definition}
       (\textbf{First Order Stochastic Dominance (FOSD)}) Let $\prior_1$,$\prior_2 \in \probspace(\statedim)$, then $\prior_1$ first order stochastically dominates  $\prior_2$ ($\prior_1 \geq_s \prior_2$) if $\sum_{\stateidx=\stateidxalt}^{\statedim} \prior_1(\stateidx) \geq \sum_{\stateidx=\stateidxalt}^{\statedim} \prior_2(\stateidx) \ \forall \ \stateidxalt\in\statespace$. 
    \end{definition}
% It can be shown that MLR is a stronger condition and $\prior_1 \geq_s \prior_2 \implies \prior_1 \geq_r \prior_2$~\cite{krishnamurthy_partially_2016}. 

The stopping time problem with the socialistic cost of~\eqref{eq:socialwelfarecost} can be decomposed into two cost terms, each corresponding to the cost terms of the stopping time problem for partially observed Markov decision processes of Theorem 12.3.4 of~\cite{krishnamurthy_partially_2016}. 
\begin{align}\label{eq:bigcost}
\begin{split}
    \bigcost(\prior,1) &= \frac{1}{1-\discountfactor} \min_{\action} \cost_\action^{\prime} \prior\\
    \bigcost(\prior,2) &= \sum_{\observation\in\observationspace} \cost_{\observation}^\prime \observationmatrix_\observation \prior + (\delaycost + (1-\discountfactor)\errorcost)\indicatorstate_1\prior - (1-\discountfactor)\errorcost
\end{split}
\end{align}
$\bigcost(\prior,1)$ is the expected cost after the decision to herd has been made. Similarly, the first term in $\bigcost(\prior,2)$ is the expected cost when revealing private observations. The rest of the terms come from transforming the value function by the delay penalty costs~\cite{krishnamurthy_bayesian_2011-1}. 

   We now state the main assumptions of Theorem~12.3.4 of~\cite{krishnamurthy_partially_2016}, which are required for this to hold for the structural result of~\eqref{eq:optpolicy},
    \begin{enumerate}
        \item \textbf{(C)} $\prior_1 \geq_s \prior_2$ implies $\bigcost(\prior_1,\action) \leq \bigcost(\prior_2,\action)$.
        \item \textbf{(F1)} $\observationmatrix$ is totally positive of order 2 (TP2).
        \item \textbf{(F2)} $\markovkernel$ is totally positive of order 2 (TP2)
        \item (\textbf{S}) $\bigcost(\prior,\decision)$ is submodular on $[\lines(\indicatorstate_\statedim,\bar{\prior}),\geq_{\linesymb_\statedim}]$ and  $[\lines(\indicatorstate_1,\bar{\prior}),\geq_{\linesymb_1}]$.
    \end{enumerate}
    F1 follows from S4, and F2 follows from the fact that we only consider an identity transition matrix $\markovkernel$. And since the costs are linear, (C) follows by applying the definition of FOSD on equation~\eqref{eq:bigcost} and using assumption (S1). (S) follows from the definition of MLR ordering on the line, using the fact that $\observationmatrix$ is TP2 in the first term of $\bigcost(\prior,2)$ and using (S2) and (S3). Hence, the assumptions are verified, and the structural result is proved.
    \end{proof}
    \subsection{Brief Experimental Details}\label{sec:expapp}
\subsubsection{Hyperparameters Of LLM at Inference Time}
As described in the main text, we use a Mixtral-8x7B-v0.1. We consider the maximum response tokens as 100, a temperature of 0.7, a top-p of 0.7, a top-k of 50, and a repetition penalty as 50. We use the TogetherAI API to send the prompt to the LLM and receive the response. 
    \subsubsection{System Prompt}
    We consider the following system prompt,
    
    \texttt{[INST]\\  Return a JSON with the following format for the given text:\\ \{`is\_insulting': Bool,\\`is\_dehumanizing':Bool,\\`is\_humiliating':Bool,\\`promotes\_violence':Bool,\\`promotes\_genocide':Bool,\\`is\_respectful':Bool\}\\ 
    Text: \{\textit{comment}\}[/INST]},
    where ``\textit{comment}'' contains the comment observation which needs to be analyzed. 
\subsubsection{Response}
In the $240$ comment observations we parsed, the response of the LLM includes the JSON response at the start and an explanation for the corresponding mapping. We truncated the output to include the JSON and got a discrete low-dimensional observation from the textual comment. 
  \subsubsection{Reducing the observation space} 
Although the LLM output for text observation is of cardinality $2^6 = 64$, we reduce by considering the following order: respectful < insulting < dehumanizing < humiliating < violence < genocide~\cite{sachdeva_measuring_2022}. The binary map $\psi(z):\{0,1\}^6 \to 6 = \max\{ \stateidxalt:  \text{s.t.} \ z[\stateidxalt] = 1\} $ takes the observation as the highest severity present in the binary observation. 
\subsubsection{Likelihood Neural Network (Restricted Boltzmann Machine) } 
We use a subset of the labeled dataset (size = $240$ comments) to obtain  
We use $|\statedim| = 6$ restricted Boltzmann Machines to approximate the likelihood function $\probabilitymeasure(\observation|\statevar)$; we train each machine on observations obtained from different states (as defined in the main text). Each RBM has $6$ visible units and $4$ hidden units. We train the RBM using contrastive divergence for $100$ epochs and generate $1000$ samples using Gibbs sampling with $1000$ iterations. We obtain the approximate probabilities by empirically counting the samples. 
\end{document}